\documentclass[option]{elsarticle}
\usepackage{amsmath,mathtools,amssymb}
\allowdisplaybreaks%允许公式分页
\usepackage{graphicx}  
\usepackage{subfigure}
\usepackage{color}
\usepackage{verbatim} 
\usepackage{booktabs}
\usepackage{bm}  
\usepackage{dutchcal}
\usepackage[left=3cm,right=3cm,top=2.5cm,bottom=2.5cm]{geometry} 
\usepackage{lineno,hyperref}
\usepackage{tablists} 
\usepackage{makecell} 
\usepackage{tikz}
\usepackage{hyperref}
\usepackage{pstricks}
\usepackage[titletoc]{appendix}
\usepackage{pstricks-add}
\usepackage{float}
\usepackage{mathrsfs} 
\usepackage{amsthm}
\usepackage{xcolor}
\biboptions{numbers,sort&compress}
\hypersetup{
    colorlinks=true,
    linkcolor=blue,
    filecolor=gray,
    urlcolor=blue,
    citecolor=blue}

\numberwithin{equation}{section}

\newtheorem{theorem}{Theorem}[section]
\newtheorem{lemma}{Lemma}[section]	
\newtheorem{proposition}{Proposition}[section]
\newtheorem {remark} {Remark}

\begin{document} %开始正文书写

    \begin{frontmatter} %开始组织Front Matter
        
        \title{Chern-Simons Type Characteristic Classes \\of  Abelian Lattice Gauge Theory} %论文题目
        
        \author{Mengyao Wu\corref{cor1}}
        \ead{mengyao_wu_w@163.com}
        \author{Jie Yang}
        \ead{5972@cnu.edu.cn}
        \cortext[cor1]{Corresponding author.}
        \address{School of Mathematical Sciences, Capital Normal University, Beijing 100048, China}
        \date{}

		\begin{abstract} 
   In this paper, we extend the definition of the Chern-Simons type characteristic classes in the continuous case to Abelian lattice gauge theory. Then, we show that the exterior differential of the $k$-th Chern-Simons type characteristic class is exactly equal to the coboundary of the cochain of the $(k-1)$-th Chern-Simons type characteristic class based upon the noncommutative differential calculus on the lattice.

\end{abstract}

\end{frontmatter}

%\tableofcontents

	\section{Introduction}
Chern-Simons secondary characteristic classes have proven to be powerful tools in studying the global differential geometry of manifolds with boundary \cite{chern1974characteristic}. In 1985, Guo et al. \cite{han1985chern} generalized the characteristic polynomials $P(F^r)$ to a sequence of new characteristic classes, called Chern-Simons type characteristic classes. They explored the relationships between these classes and derived the descent (or zig-zag) equations. Subsequently, Guo et al. \cite{han1985anomalies} investigated the applications of non-Abelian anomalies and the gauge-invariant Wess-Zumino effective action \cite{chou2009gauge,wess1971consequences, witten1983global} by combining the cohomology groups of gauge transformations with the Chern-Simons type characteristic classes. From then on, descent equations have been widely used in the literature \cite{Chou_1985, Chou_1984, chou1984symmetric} to learn the anomaly. In light of its importance, later studies have sought to generalize descent equations \cite{kang2018descent,alekseev2018chern,izaurieta2015chern,izaurieta2017chern,danhua2024higher}.
   
 In the continuous gauge theory, the Chern-Simons type characteristic classes are constructed from $k+1$ Lie algebra $\mathcal{g}$-valued 1-forms $A_i,i \in \{0,\cdots, k\}$ and a $k$-simplex set $\Delta_k$. With the  interpolation $A_{0,t_1 \cdots t_k}=A_0+\sum\limits_{i=1}^{k} t_i (A_i-A_0)=A_0+\sum\limits_{i=1}^{k} t_i \eta^{i,0}$ and the corresponding curvature $F_{0,t_1 \cdots t_k}$, the associated $(2r-k)$-form $Q_r^{(k)}$ is defined as
     \begin{align} \label{Q-p}
         Q_r^{(k)}(A_0,\cdots,A_i,\cdots A_k;\Delta_k)=\frac{r!}{(r-k)!} \int_{\Delta_{k}} dt_1 \wedge \cdots \wedge dt_k P(\eta^{1,0},\cdots,\eta^{k,0},F_{0,t_1 \cdots t_k}^{r-k}),0\leq k \leq r,
     \end{align}
   where $P$ is a $G$-invariant form  (see Ref \cite{han1985chern}). Then, the descent equations are given as follows: 
   \begin{align}
   \Delta Q_r^{(k-1)}(A_0,\cdots,\hat{A_i},\cdots A_k;\partial \Delta_k)= dQ_r^{(k)}(A_0,\cdots,A_i,\cdots A_k;\Delta_k).
   \end{align}
 
Despite the significant achievements of Chern-Simons theory in the continuous case, there are several unclear issues within the continuous framework, such as the framing anomaly. However, these problems can be effectively addressed in the lattice gauge theory \cite{gromov2015framing}. In turn, the lattice gauge theory provides novel methods for resolving many of the ambiguities present in the continuous field theory. The lattice gauge theory performs calculations on a discretized space-time grid, replacing the continuous space-time background. Unlike traditional differential geometry methods in the continuous case, noncommutative differential geometry (NCDG) offers a more suitable framework for discrete situations. Several studies have explored the unique features of lattice gauge theory, such as the expression of the topological charge (Chern classes), the properties of chiral anomalies and the linking number on the lattice \cite{niedermayer1999exact,luscher1999topology, fujiwara2001application,zhang2022abelian}. Notably, T. Fujiwara et al. \cite{fujiwara2001application} derived the chiral anomaly in the Abelian lattice gauge theory by solving the descent equations and highlighted the topological significance of Chern classes within the framework of NCDG.
    
  Given that the descent equations play a crucial role in the continuous gauge theory, this paper aims to generalize the Chern-Simons type characteristic classes to the lattice settings and derive the descent equations for the lattice. This exploration helps in understanding the Chern-Simons theory and anomalies in the lattice context. The NCDG framework used in here is based on the simplest case of noncommutative differential calculus on discrete sets. 
 Based on the Abelian lattice gauge theory, we define a sequence of Chern-Simons type characteristic classes and demonstrate an important property: the exterior differential of the $k$-th Chern-Simons type characteristic class is exactly equal to the coboundary of the cochain of the $(k-1)$-th Chern-Simons type characteristic class on the lattice.
 
  The paper is structured as follows. In Section \ref{section 2},  we review the relevant topics in the Abelian lattice gauge theory and noncommutative differential calculus on hypercubic lattice. In Section \ref{section 3}, we define $G$-invariant form on the lattice and prove the exchange law for the differential forms on the lattice. In section \ref{section 4}, we generalize the definitions of the Chern-Simons type characteristic classes to the Abelian lattice gauge theory and  prove the validity of the descent equations for the lattice. In the continuous limit, the form reduces to the same form in \cite{han1985chern}.
  
\section{Noncommutative differential calculus in Abelian lattice gauge  theory} \label{section 2}
In this section, we review some useful definitions in the Abelian lattice gauge theory and some basic results. Then, we introduce the noncommutative differential calculus on the lattice.  For simplicity, we consider the lattice with unit spacing, i.e., $ h = 1$. More details can be found in \cite{dimakis1993noncommutative,fujiwara2000non,sulejmanpasic2019abelian,fujiwara2001topological}.

Let $P=P(M,G)$ be a principal $G$-bundle over a $D$-dimensional base manifold $M$, where $G$ is an Abelian Lie group. We discretize the manifold $M$ into unit cells $c(n)(n \in \mathbb{Z}^{D})$,
\begin{align}
    c(n)=\left\{x \in  R^D \vert 0 \leqslant(x_{\mu}-n_{\mu})\leqslant 1,  \forall \mu \in \{1,\cdots, D\} \right\}.
\end{align}

Denote the lattice by $\Lambda$, which serves as the discrete base space. For a node $ x \in \Lambda$, the fiber $ G_x = \pi^{-1} (x)$ is isomorphic to the structure group $G$. All of these fibers form a discrete principal bundle, denoted $ Q(\Lambda, G) $. 

Let $\mathcal{A}(\mathbb{Z}^D)$ be the algebra of functions defined on $\mathbb{Z}^D$. The forward difference operator $\Delta_{\mu}$ and the backward difference operator $\Delta_{\mu}^*$ correspond to the ordinary derivative $\partial_x$ on $\mathcal{A}(M)$, where $\mathcal{A}(M)$ is the algebra of functions defined on the manifold $M$.
\begin{align}
    \Delta_{\mu}f(n)=f(n+\hat{\mu})-f(n),\nonumber\\
    \Delta_{\mu}^* f(n)=f(n)-f(n-\hat{\mu}).
     \end{align}
 
The bases of the 1-form on a $D$-dimensional hypercubic lattice are $  dx_1,dx_2,\cdots,dx_D$, defined on the links and satisfy the Grassmann algebra:
 \begin{align}
     dx_{\mu} \wedge dx_{\nu}=-dx_{\nu} \wedge dx_{\mu}.
 \end{align}
We define a $k$-form ${\omega^{(n)}(x)}$ on the cell $c(n)$  i.e. $x \in c(n)$ as follows:
\begin{align}
    \omega^{(n)}(x)=\frac{1}{k!}\sum_{\mu_1,\cdots,\mu_k} \omega_{\mu_1,\cdots,\mu_k}^{(n)}(x)dx_{\mu_1}\wedge \cdots\wedge dx_{\mu_k},
\end{align}
where the  coefficient $\omega_{\mu_1,\cdots,\mu_k}^{(n)}(x) $ is an antisymmetric tensor of rank $k$. The exterior derivative operator $d$  is defined by the forward difference operator as:
\begin{align}
    d\omega^{(n)}(x)=\frac{1}{k!}\sum_{\mu,\mu_1,\mu_2,\cdots,\mu_k}\Delta_{\mu}
    \omega _{\mu_1,\mu_2,\cdots,\mu_k}^{(n)}(x) dx_{\mu}\wedge dx_{\mu_1}\wedge \cdots\wedge  dx_{\mu_k}.
\end{align}
$d$ is nilpotent, i.e., $d^2=0$.

The essence of noncommutative differential calculus \cite{connes1994noncommutative, hanying2000noncommutative} is the following
relation:
\begin{align}\label{NCDC}
    dx_{\mu}f(x)=f(x+\hat{\mu})dx_{\mu},
\end{align}
 where $f(x)$ is a function. Specifically, a function defined on the lattice does not commute trivially with the differential form; instead, the exchange results in a shift of the coordinate. This represents a fundamental distinction between the continuous and discrete cases. Notably, the key observation is that the Leibniz rule for the exterior derivative $d$ remains valid.
\begin{align}
    d[\omega^{(n)}(x) \theta ^{(n)}(x)]=d\omega^{(n)}(x)\theta ^{(n)}(x)+(-1)^{k}\omega^{(n)}(x)d\theta ^{(n)}(x),
\end{align}
where $\omega^{(n)}(x)$ is a $k$-form.

The discrete connection 1-form and the field strength 2-form are
\begin{align}
    A^{(n)}(x)&=\sum_{\mu}A_{\mu}^{(n)}(x)dx_{\mu},\\
    F^{(n)}(x)&=\frac{1}{2}\sum_{\mu,\nu} F^{(n)}_{\mu\nu}(x)dx_{\mu}dx_{\nu}=\frac{1}{2}\sum_{\mu,\nu} (\Delta_{\mu}A_{\nu}^{(n)}-\Delta_{\nu}A_{\mu}^{(n)})dx_{\mu}dx_{\nu}=dA^{(n)}(x).
\end{align}
In the continuous case, the connection 1-form takes value in a Lie algebra. Here, we consider the 1-form $A^{(n)}(x)$ is real valued. The Bianchi identity  for the difference discrete curvature is
\begin{align} \label{Bianchi}
    dF^{(n)}(x)=0 \Longleftrightarrow \sum_{\lambda,\rho,\sigma}\epsilon^{\lambda \rho \sigma}\Delta_{\lambda}F_{\rho,\sigma}^{(n)}(x)=0.
\end{align}

The gauge transformation of the discrete connection 1-form \cite{wu2006difference} is 
\begin{align}\label{a}
   A_{\mu}^{(n)}(x)\longmapsto A_{\mu}^{(n)}(x)-\Delta_{\mu} \lambda^{(n)}(x).
\end{align}
It is  easy to check  that the curvature
under the gauge transformation in \eqref{a} is simply gauge invariant:
\begin{align}
    F_{\mu\nu}^{(n)}(x) \longmapsto  F_{\mu\nu}^{(n)}(x).
\end{align}

\section{\texorpdfstring{$G$}{G}-invariant form on the lattice} \label{section 3}
Similar to the role of the $G$-invariant polynomial in the continuous descent equations, we define the $G$-invariant form in the Abelian lattice  gauge theory.
Assume $P$ is a $G$-invariant form of degree $r$, namely $P$ is a symmetric, multilinear, $G$-invariant map
\begin{align*}
    P:\underbrace{\mathcal{A}(Z^D) \times \cdots \times \mathcal{A}(Z^D)}_{r} \longrightarrow \mathbb{R},
\end{align*}
satisfying
\begin{align}
     P(\alpha_1,\cdots,\alpha_i,\cdots,\alpha_j,\cdots,\alpha_r)&=P(\alpha_1,\cdots,\alpha_j,\cdots,\alpha_i,\cdots,\alpha_r),
    \end{align}
\begin{align}\label{G-in}
    P(g \alpha_1 g^{-1},\cdots,g \alpha_r g^{-1})&=P( \alpha_1 ,\cdots, \alpha_r ),
\end{align}
for all $\alpha_1,\cdots,\alpha_r \in \mathcal{A}(Z^D)$, $g \in G $.

We generalize the definition of the action of $P$ to $\mathcal{A}(Z^D)^{d_1} \times \cdots \times \mathcal{A}(Z^D)^{d_r}$, in which $\mathcal{A}(Z^D)^{d_i}$ denotes the vector space of real valued $d_i$-forms on the lattice $\Lambda$.  For example, when $W_i^{(n)}$ are  $d_i$-forms on the cell $c(n)$, written as $W_i^{(n)}= w_{\mu_{i_1},\cdots,\mu_{i_{d_i}}}^{(n)} dx_{\mu_{i_1}}\wedge \cdots \wedge dx_{\mu_{i_{d_i}}}$ \footnote{To simplify the notation, we adopt the Einstein summation convention in this section, where summation is implied over every pair of repeated indices.}, we define
\begin{align}
    P(W_1,\cdots,W_r)&=P(w_{\mu_{1_1},\cdots,\mu_{1_{d_1}}}^{(n)},\cdots, w_{\mu_{r_1},\cdots,\mu_{r_{d_r}}}^{(n+\mu_{1_1}+\cdots+\mu_{1_{d_1}}+\cdots+\mu_{{(r-1)}_1}+\cdots+\mu_{{(r-1)}_{d_{r-1}}})})\times  \\\nonumber 
    &dx_{\mu_{1_1}}\wedge \cdots \wedge dx_{\mu_{1_{d_1}}} \wedge \cdots \wedge dx_{\mu_{r_1}}\wedge \cdots \wedge dx_{\mu_{r_{d_r}}}
\end{align}
and we have
\begin{align}
  P(g W_1^{(n)} g^{-1},\cdots,g  W_r^{(n)}  g^{-1})=P( W_1^{(n)},\cdots,  W_r^{(n)}),
\end{align}
by using \eqref{G-in}.

The following proposition provides the commutative law for real valued differential forms on the lattice, which corresponds to the commutative law for Lie algebra valued differential forms in the $G$-invariant polynomial $P$ in the continuous case.
\begin{proposition}
    Define a $k$-form on the cell $c(n)$ as $W_1^{(n)}=w_{1\mu_1 ,\cdots ,\mu_k}^{(n)}dx_{\mu_1} \cdots dx_{\mu_k}$, and an $l$-form on the cell $c(n)$ as $W_2^{(n)}=w_{2\nu_1 ,\cdots ,\nu_l}^{(n)} dx_{\nu_1} \cdots dx_{\nu_l}$. We have
   \begin{align}
       P(W_1^{(n)},W_2^{(n)})
      =(-1)^{kl}P(w_{2 \nu_1 ,\cdots ,\nu_l}^{(n+\mu_1 \cdots +\mu_k)} dx_{\nu_1} \wedge\cdots\wedge dx_{\nu_l} , w_{1 \mu_1 ,\cdots ,\mu_k}^{(n-\nu_1 -\cdots -\nu_l)}  dx_{\mu_1} \wedge \cdots \wedge dx_{\mu_k}).
    \end{align}
    \end{proposition}
       \begin{proof}
          \begin{align*}
            &P(W_1^{(n)},W_2^{(n)})\\
       &=P(w_{1 \mu_1 ,\cdots ,\mu_k}^{(n)}dx_{\mu_1} \wedge \cdots \wedge dx_{\mu_k} ,w_{2 \nu_1 ,\cdots ,\nu_l}^{(n)}\wedge dx_{\nu_1} \wedge \cdots \wedge dx_{\nu_l}) \\
       &=P(w_{1 \mu_1 ,\cdots ,\mu_k}^{(n)}, w_{2 \nu_1 ,\cdots ,\nu_l}^{(n+\mu_1 \cdots +\mu_k)})dx_{\mu_1} \wedge \cdots \wedge dx_{\mu_k} \wedge  dx_{\nu_1}\wedge \cdots \wedge  dx_{\nu_l} \\
       &=(-1)^{kl}P(w_{2 \nu_1 ,\cdots ,\nu_l}^{(n+\mu_1 \cdots +\mu_k)}, w_{1 \mu_1 ,\cdots ,\mu_k}^{(n)}) dx_{\nu_1} \wedge\cdots \wedge dx_{\nu_l} \wedge dx_{\mu_1} \wedge \cdots \wedge dx_{\mu_k} \\
       &=(-1)^{kl}P(w_{2 \nu_1 ,\cdots ,\nu_l}^{(n+\mu_1 \cdots +\mu_k)} dx_{\nu_1} \wedge\cdots\wedge dx_{\nu_l} , w_{1 \mu_1 ,\cdots ,\mu_k}^{(n-\nu_1 -\cdots -\nu_l)}  dx_{\mu_1} \wedge \cdots \wedge dx_{\mu_k}) 
       \end{align*}
by using \eqref{NCDC}.
\end{proof}

 In contrast to the continuous case, this proposition demonstrates that the commutation of differential forms on the lattice, within the invariant polynomial $P$, is accompanied by a shift in the lattice positions. Therefore,
\begin{align}
    P( W_1^{(n)},\cdots,  W_i^{(n)},\cdots, W_j^{(n)},\cdots, W_r^{(n)} )&\neq (-1)^{(d_{i+1}+\cdots+d_{j-1})(d_i +d_j)+d_i d_j}\nonumber
    \\ &P( W_1^{(n)},\cdots,  W_j^{(n)},\cdots, W_i^{(n)},\cdots, W_r^{(n)}).
\end{align}

Given the properties of the lattice gauge theory, we emphasize that the positions of differential forms will not be arbitrarily altered in the subsequent proofs. This ensures that all operations remain confined within the lattice framework $c(n)$.

\section{Chern-Simons Type Characteristic Classes and Descent Equations on the Lattice} \label{section 4}
 In this section, we construct the Chern-Simons type characteristic classes and derive the descent equations based on the Abelian lattice gauge theory. The key issue in this process is the application of noncommutative differential calculus. 
 
 As shown in \eqref{Q-p} \cite{han1985chern}, the integration is performed over $k$-simplex set $\Delta_k$ consisting of $  0 \leq t_i \leq 1 (i=1,\cdots k),\sum \limits_{i=1}^{k} t_i \leq 1$ in the continuous case. In the Abelian gauge theory, we focus on the discretized $k$-simplex set, consisting of vertices and the edges linking them. For simplicity, we denote this discretized set as $\Delta_k$. 
 
 Given $k+1$ discrete connections $A_0,\cdots A_k$ on the discrete principal bundle $Q(\Lambda,G)$, we define the Chern-Simons type characteristic classes on the cell $c(n)$ as follows: 
 \begin{align} \label{Q}
     Q_r^{(k)}(A_0^{(n)},\cdots,A_i^{(n)},\cdots A_k^{(n)};\Delta_k)=\sum_{\sum\limits_{l=0}^{k} m_l=r-k}P(F_0^{m_0},\eta^{1,0},\cdots,F_{i-1}^{m_{i-1}},\eta^{i,i-1},F_i^{m_i},\cdots,\eta^{k,k-1},F_k^{m_k}).
 \end{align}
It is a $2r-k$-form, where $F_i=dA_i^{(n)}, \eta^{i,i-1}=A_i^{(n)}-A_{i-1}^{(n)}$, and $m_i \in \mathbb{N}$, with $0 \leq k \leq r.$ 
 
The following lemma provides a crucial step towards proving Lemma \ref{Lemma 4.2}.
\begin{lemma}\label{doc}
    \begin{align}\label{3.2}
        P(\eta^{i,j},\eta^{k,i})-P(\eta^{l,j},\eta^{k,l})=P(\eta^{i,j},\eta^{l,i})-P(\eta^{l,i},\eta^{k,l}).
    \end{align}
    \end{lemma}
\begin{proof}
    \begin{align*}
      P(\eta^{i,j},\eta^{k,i})-P(\eta^{l,j},\eta^{k,l})
      &=P(\eta^{i,j},\eta^{k,i})-P(\eta^{i,j},\eta^{k,l})+P(\eta^{i,j},\eta^{k,l})-P(\eta^{l,j},\eta^{k,l})
     \\&=P(\eta^{i,j},\eta^{l,i})-P(\eta^{l,i},\eta^{k,l}),
 \end{align*}
by the multilinear of $P$.
\end{proof}

To simplify the proof of \eqref{THE}, we place the final step of the derivation into the proof of the subsequent lemma. Straightforward calculations prove that:
\begin{lemma}
    \label{Lemma 4.2}
   \begin{align*}
        &\sum\limits_{j=2}^{k}(-1)^{j-1}\sum_{\sum\limits_{l=0}^{k} m_l^{\widehat{j-2,j-1}}=r-k}P(F_0^{m_0},\cdots,F_{j-3}^{m_{j-3}},\eta^{j-2,j-3},\eta^{j-1,j-2},F_j^{m_j+1},\cdots,F_k^{m_k})
        \\&+\sum\limits_{i=1}^{k-1}(-1)^{i}\sum_{\sum\limits_{l=0}^{k} m_l^{\widehat{i,i+1}}=r-k}P(F_0^{m_0},\cdots,F_{i-1}^{m_{i-1}+1},\eta^{i+1,i},\eta^{i+2,i+1},F_{i+2}^{m_{i+2}},\cdots,F_k^{m_k})
        \\&=\sum_{\sum\limits_{l=0}^{k} m_l^{\widehat{0,1}}=r-k+1}P(\eta^{2,1},F_{2}^{m_{2}},\eta^{3,2},\cdots,F_k^{m_k})+(-1)^k\sum_{\sum\limits_{l=0}^{k} m_l^{\widehat{k-1,k}}=r-k+1} P(F_0^{m_0},\cdots,F_{k-2}^{m_{k-2}},\eta^{k-1,k-2})
        \\&+\sum\limits_{t=1}^{k-1}(-1)^{t}\sum_{\sum\limits_{l=0}^{k} m_l^{\widehat{t,t+1}}=r-k}P(F_0^{m_0},\cdots,F_{t-1}^{m_{t-1}+1},\eta^{t+1,t-1},\eta^{t+2,t+1},F_{t+2}^{m_{t+2}},\cdots,F_k^{m_k})
        \\&+\sum\limits_{s=1}^{k-1}(-1)^{s}\sum_{\sum\limits_{l=0}^{k} m_l^{\widehat{s-1,s}}=r-k+1}P(F_0^{m_0},\cdots,F_{s-2}^{m_{s-2}},\eta^{s-1,s-2},\eta^{s+1,s-1},F_{s+1}^{m_{s+1}},\cdots,F_k^{m_k}),
   \end{align*}
where $m_{l}^{\hat{i}}$ means $\mathop{m_{l}|}\nolimits_{l=i}=0$.

\end{lemma}
\begin{proof}
   Firstly, we apply Lemma \ref{doc}, then combine like terms on both sides of the equation and cancel out the common terms to simplify the above equation.
   Thus, the equation in Lemma \ref{Lemma 4.2} is equal to 
   \begin{align}\label{3.3}
       &\sum\limits_{j=2}^{k-1}(-1)^{j}\sum_{\sum\limits_{l=0}^{k} m_l^{\widehat{j-2,j-1,j}}=r-k}P(F_0^{m_0},\cdots,F_{j-3}^{m_{j-3}},\eta^{j-2,j-3},\cdots,\eta^{j,j-1},F_{j+1}^{m_{j+1}+1},\cdots,F_k^{m_k})\nonumber
       \\&=-\sum_{\sum\limits_{l=0}^{k} m_l^{\widehat{0,1,2}}=r-k+1}P(\eta^{1,0},\eta^{3,2},F_{3}^{m_{3}},\cdots,F_k^{m_k})\nonumber
       \\&+(-1)^k\sum_{\sum\limits_{l=0}^{k} m_l^{\widehat{k-2,k-1,k}}=r-k+1} P(F_0^{m_0},\cdots,F_{k-3}^{m_{k-3}},\eta^{k-2,k-3},\eta^{k-1,k-2})\nonumber
       \\&+\sum\limits_{t=2}^{k-1}(-1)^{t}\sum_{\sum\limits_{l=0}^{k} m_l^{\widehat{t-1,t,t+1}}=r-k}P(F_0^{m_0},\cdots,F_{t-2}^{m_{t-2}+1},\eta^{t-1,t-2},\eta^{t,t-1},\eta^{t+2,t+1},F_{t+2}^{m_{t+2}},\cdots,F_k^{m_k})\nonumber
       \\&+\sum\limits_{s=0}^{k-3}(-1)^{s}\sum_{\sum\limits_{l=0}^{k} m_l^{\widehat{s,s+1,s+2}}=r-k+1}P(F_0^{m_0},\cdots,F_{s-1}^{m_{s-1}},\eta^{s,s-1},\eta^{s+1,s},\eta^{s+3,s+1},F_{s+3}^{m_{s+3}},\cdots,F_k^{m_k}) .    
   \end{align}
Similar procedure gives rise to
   \begin{align}\label{3.4}
    &\sum\limits_{j=3}^{k-1}(-1)^{j}\sum_{\sum\limits_{l=0}^{k} m_l^{\widehat{j-3,\cdots ,j}}=r-k}P(F_0^{m_0},\cdots,F_{j-4}^{m_{j-4}},\eta^{j-3,j-4},\cdots,\eta^{j,j-1},F_{j+1}^{m_j+1},\cdots,F_k^{m_k})\nonumber
   \\&=\sum_{\sum\limits_{l=0}^{k} m_l^{\widehat{0,\cdots,3}}=r-k+1}P(\eta^{1,0},\eta^{2,1},\eta^{4,3},F_{4}^{m_{4}},\cdots,F_k^{m_k})\nonumber
   \\&+(-1)^k\sum_{\sum\limits_{l=0}^{k} m_l^{\widehat{k-3,\cdots ,k}=r-k+1}} P(F_0^{m_0},\cdots,F_{k-4}^{m_{k-4}},\eta^{k-3,k-4},\cdots,\eta^{k-1,k-2})\nonumber
   \\&+\sum\limits_{t=3}^{k-1}(-1)^{t}\sum_{\sum\limits_{l=0}^{k} m_l^{\widehat{t-2,\cdots ,t+1}}=r-k}P(F_0^{m_0},\cdots,F_{t-3}^{m_{t-3}+1},\eta^{t-2,t-3},\eta^{t-1,t-2},\eta^{t,t-1},\eta^{t+2,t+1},F_{t+2}^{m_{t+2}},\cdots,F_k^{m_k})\nonumber
   \\&+\sum\limits_{s=1}^{k-3}(-1)^{s}\sum_{\sum\limits_{l=0}^{k} m_l^{\widehat{s-1,\cdots ,s+2}}=r-k+1}P(F_0^{m_0},\cdots,F_{s-2}^{m_{s-2}},\eta^{s-1,s-2},\eta^{s,s-1},\eta^{s+1,s},\eta^{s+3,s+1},F_{s+3}^{m_{s+3}},\cdots,F_k^{m_k})  .   
\end{align}
The calculation process above is repeated until $j$ starts from $q$ in the summand on the left-hand side. Then, the equation \eqref{3.4} simplifies to the following:
 \begin{align}\label{3.5}
    &\sum\limits_{j=q}^{k-1}(-1)^{j}\sum_{\sum\limits_{l=0}^{k} m_l^{\widehat{j-q,\cdots ,j}}=r-k}P(F_0^{m_0},\cdots,F_{j-q-1}^{m_{j-q-1}},\eta^{j-q,j-q-1},\cdots,\eta^{j,j-1},F_{j+1}^{m_j+1},\cdots,F_k^{m_k})\nonumber
    \\&=(-1)^{q-1}\sum_{\sum\limits_{l=0}^{k} m_l^{\widehat{0,\cdots,q}}=r-k+1,}P(\eta^{1,0},\eta^{2,1},\cdots,\eta^{q-1,q-2},\eta^{q+1,q},F_{q+1}^{m_{q+1}},\cdots,F_k^{m_k})\nonumber
    \\&+(-1)^k\sum_{\sum\limits_{l=0}^{k} m_l^{\widehat{k-q,\cdots,k}}=r-k+1} P(F_0^{m_0},\cdots,F_{k-q-1}^{k-q-1},\eta^{k-q,k-q-1},\cdots,\eta^{k-1,k-2})\nonumber
    \\&+\sum\limits_{t=q}^{k-1}(-1)^{r}\sum_{\sum\limits_{l=0}^{k} m_l^{\widehat{t-q+1,\cdots ,t+1}}=r-k}P(F_0^{m_0},\cdots,F_{t-q}^{m_{t-q}+1},\eta^{t-q+1,t-q},\cdots,\eta^{t,t-1},\eta^{t+2,t+1},F_{t+2}^{m_{t+2}},\cdots,F_k^{m_k})\nonumber
    \\&+\sum\limits_{s=q-2}^{k-3}(-1)^{s}\sum_{\sum\limits_{l=0}^{k} m_l^{\widehat{s-q+2,\cdots ,s+2}}=r-k+1}P(F_0^{m_0},\cdots,F_{s-q+1}^{m_{s-q+1}},\eta^{s-q+2,s-q+1},\cdots,\eta^{s+1,s},\eta^{s+3,s+1},F_{s+3}^{m_{s+3}},\cdots,F_k^{m_k}) .
\end{align}
When $q=k-1$, the equation\eqref{3.5} is reduced to 
\begin{align}\label{3.6}
    &(-1)^{k-1}P(\eta^{1,0},\eta^{2,1},\cdots,\eta^{k-1,k-2},F_k^{r-k+1})
    \nonumber\\ 
    &=(-1)^{k}P(\eta^{1,0},\eta^{2,1},\cdots,\eta^{k,k-1},F_{k}^{r-k+1})
    +(-1)^{k}P(F_0^{r-k+1},\eta^{1,0},\cdots,\eta^{k-1,k-2})\nonumber
    \\&+(-1)^{k-1}P(F_0^{r-k+1},\eta^{1,0},\cdots,\eta^{k-1,k-2})+(-1)^{k-1}P(\eta^{1,0},\eta^{2,1},\cdots,\eta^{k,k-2},F_{k}^{r-k+1}).
\end{align}
Obviously, the equation \eqref{3.6} holds.
\end{proof}

\begin{remark}
 The cancellation process on both sides of the equation involved in the above calculations is based on the following fact:
\begin{align}
   &\sum_{\sum\limits_{l=0}^{k} m_l=r-k} P(F_0^{m_0},\cdots,\eta^{i-1,i-2},F_{i-1}^{m_{i-1}},F_i^{m_i+1},\eta^{i+1,i},\cdots,F_k^{m_k})\\\nonumber
   &-\sum_{\sum\limits_{l=0}^{k} m_l=r-k} P(F_0^{m_0},\cdots,\eta^{i-1,i-2},F_{i-1}^{m_{i-1}+1},F_i^{m_i},\eta^{i+1,i},\cdots,F_k^{m_k})\\\nonumber
   &=\sum_{\sum\limits_{l=0}^{k} m_l^{\widehat{i-1}}=r-k} P(F_0^{m_0},\cdots,\eta^{i-1,i-2},F_{i-1}^{m_{i-1}},F_i^{m_i+1},\eta^{i+1,i},\cdots,F_k^{m_k})\\\nonumber
   &-\sum_{\sum\limits_{l=0}^{k} m_l^{\widehat{i}}=r-k} P(F_0^{m_0},\cdots,\eta^{i-1,i-2},F_{i-1}^{m_{i-1}+1},F_i^{m_i},\eta^{i+1,i},\cdots,F_k^{m_k}).
\end{align}
\end{remark}
 Having established a sequence of Chern-Simons characteristic classes $Q_r^{(k)}$ and $G$-invariant form identity in Lemma \ref{Lemma 4.2}, we present a remarkable property of the $Q$-polynomials in the following theorem:
\begin{theorem} 
   If the $Q$-polynomials are defined as \eqref{Q}, in which $P$ is a $G$-invariant form on the lattice $\Lambda$, then there exists a relation between 
   a $k$-th Q-polynomial and the $(k-1)$-th Q-polynomials as follows:
    \begin{align}\label{THE}
        dQ_r^{(k)}(A_0^{(n)},\cdots,A_i^{(n)},\cdots A_k^{(n)}; \Delta_k)=\Delta Q_r^{(k-1)}(A_0^{(n)},\cdots, A_k^{(n)}; \partial \Delta_k).
    \end{align}
Here $\Delta$ is an operator acting on a polynomial $R^{(k)}(A_0^{(n)},\cdots, A_k^{(n)})$ by 
\begin{align}
    (\Delta R^{(k)})(A_0^{(n)},\cdots,A_i^{(n)},\cdots A_{k+1}^{(n)})=\sum\limits_{i=0}^{k+1}(-1)^{i} R^{(k)}(A_0^{(n)},\cdots,A_{i-1}^{(n)},\hat{A_i}^{(n)},A_{i+1}^{(n)},\cdots, A_{k+1}^{(n)}).
\end{align}
\end{theorem}

\begin{proof}
    In order to prove this theorem, we first calculate
    \begin{align*}
   &dP(F_0^{m_0},\cdots,F_{i-1}^{m_{i-1}}\eta^{i,i-1},F_i^{m_i},\cdots,F_k^{m_k})\\
   &=\sum\limits_{i=0}^{k}(-1)^{i}P(F_0^{m_0},\cdots,F_i^{l},dF_i,F_i^{m_i-l-1},\eta^{i+1,i}\cdots,F_k^{m_k})
   \\&+\sum\limits_{i=1}^{k}(-1)^{i-1}P(F_0^{m_0},\cdots,F_{i-1}^{m_{i-1}},d\eta^{i,i-1},F_i^{m_i},\cdots,F_k^{m_k})
   \\&=\sum\limits_{i=1}^{k}(-1)^{i-1}P(F_0^{m_0},\cdots,F_{i-1}^{m_{i-1}},F_i-F_{i-1},F_i^{m_i},\cdots,F_k^{m_k})
   \\&=\sum\limits_{i=1}^{k}(-1)^{i-1}P(F_0^{m_0},\cdots,\eta^{i-1,i-2},F_{i-1}^{m_{i-1}},F_i^{m_i+1},\eta^{i+1,i},\cdots,F_k^{m_k})
   \\&+\sum\limits_{i=1}^{k}(-1)^{i}P(F_0^{m_0},\cdots,\eta^{i-1,i-2},F_{i-1}^{m_{i-1}+1},F_i^{m_i},\eta^{i+1,i},\cdots,F_k^{m_k}),
\end{align*}
by using the Bianchi identity \eqref{Bianchi}.

It should be noted that we start from the above indicator of $\eta$, which is greater than or equal to zero  and  the cut-off is greater than $k$. For example, in the first term, when $i=1$, it is equals to $P(F_0^{m_0},F_1^{m_1+1},\eta^{2,1},\cdots,F_k^{m_k})$; when $i=k$, it is equals to $P(F_0^{m_0},\cdots,\eta^{k-1,k-2},F_{k-1}^{m_{k-1}},F_k^{m_k+1})$. We follow this approach in this paper.

\begin{align*}
   &dQ_r^{(k)}(A_0^{(n)},\cdots,A_i^{(n)},\cdots A_k^{(n)}; \Delta_k)\\
   &=\sum_{\sum\limits_{l=0}^{k} m_l=r-k}dP(F_0^{m_0},\cdots,F_{i-1}^{m_{i-1}},\eta^{i,i-1},F_i^{m_i},\cdots,F_k^{m_k})
   \\&=\sum\limits_{i=1}^{k}(-1)^{i-1}\sum_{\sum\limits_{l=0}^{k} m_l=r-k}P(F_0^{m_0},\cdots,F_{i-1}^{m_{i-1}},F_i^{m_i+1},\cdots,F_k^{m_k})
   \\&+\sum\limits_{i=1}^{k}(-1)^{i}\sum_{\sum\limits_{l=0}^{k} m_l=r-k}P(F_0^{m_0},\cdots,F_{i-1}^{m_{i-1}+1},F_i^{m_i},\cdots,F_k^{m_k})
   \\&=\sum\limits_{i=1}^{k}(-1)^{i-1}\sum_{\sum\limits_{l=0}^{k} m_l^{\widehat{i-1}}=r-k}P(F_0^{m_0},\cdots,\eta^{i-1,i-2},F_i^{m_i+1},\eta^{i+1,i},\cdots,F_k^{m_k})
   \\&+\sum\limits_{i=1}^{k} (-1)^{i}\sum_{\sum\limits_{l=0}^{k} m_l^{\hat{i}}=r-k} P(F_0^{m_0},\cdots,\eta^{i-1,i-2},F_{i-1}^{m_{i-1}+1},\eta^{i+1,i},\cdots,F_k^{m_k})
   \\&=\sum_{\sum\limits_{l=0}^{k} m_l^{\hat{0}}=r-k} P(F_1^{m_1+1},\eta^{2,1},\cdots F_k^{m_k})+(-1)^k\sum_{\sum\limits_{l=0}^{k} m_l^{\hat{k}}=r-k} P(F_0^{m_0},\cdots,\eta^{k-1,k-2},F_{k-1}^{m_{k-1}+1})
   \\&+\sum\limits_{s=1}^{k-1}(-1)^{s}\sum_{\sum\limits_{l=0}^{k} m_l^{\hat{s},\widetilde{s+1}}=r-k}P(F_0^{m_0},\cdots,F_{s-1}^{m_{s-1}+1},\eta^{s+1,s-1},F_{s+1}^{m_{s+1}},\cdots,F_k^{m_k})
   \\&+\sum\limits_{j=2}^{k}(-1)^{j-1}\sum_{\sum\limits_{l=0}^{k} m_l^{\widehat{j-2,j-1}}=r-k}P(F_0^{m_0},\cdots,F_{j-3}^{m_{j-3}},\eta^{j-2,j-3},\eta^{j-1,j-2},F_j^{m_j+1},\cdots,F_k^{m_k})
   \\&+\sum\limits_{i=1}^{k-1}(-1)^{i}\sum_{\sum\limits_{l=0}^{k} m_l^{\widehat{i,i+1}}=r-k}P(F_0^{m_0},\cdots,F_{i-1}^{m_{i-1}+1},\eta^{i+1,i},\eta^{i+2,i+1},F_{i+2}^{m_{i+2}},\cdots,F_k^{m_k}),
    \end{align*}
 where $m^{\widetilde{s+1}}$ means $\mathop{m_{l}|}\nolimits_{l=s+1} \neq 0$. The final step is obtained by proposing $i=1$ in the first term and $i=k$ in the second term, so that the remaining part of the first term $j=i-1$ is merged with the remaining part of the second term.

Based on the definition of $Q$-polynomials and the action of the operator $\Delta$, we can obtain
\begin{align*}
    &\Delta Q_r^{(k-1)}(A_0^{(n)},\cdots, A_k^{(n)};\partial \Delta_k)\\
    &=\sum\limits_{t=0}^{k}(-1)^{t} Q_r^{(k-1)}(A_0^{(n)},\cdots,A_{t-1}^{(n)},\hat{A_t}^{(n)},A_{t+1}^{(n)},\cdots, A_k^{(n)})
    \\&=\sum\limits_{t=0}^{k}(-1)^{t}\sum_{\sum\limits_{l=0}^{k} m_l^{\hat{t}}=r-k+1}P(F_0^{m_0},\cdots,F_{t-1}^{m_{t-1}},\eta^{t+1,t-1},F_{t+1}^{m_{t+1}},\cdots,F_k^{m_k})
    \\&=\sum_{\sum\limits_{l=0}^{k} m_l^{\hat{0}}=r-k+1}P(F_1^{m_1},\eta^{2,1},F_{2}^{m_{2}},\cdots,F_k^{m_k})+(-1)^k\sum_{\sum\limits_{l=0}^{k} m_l^{\hat{k}}=r-k+1} P(F_0^{m_0},\cdots,\eta^{k-1,k-2},F_{k-1}^{m_{k-1}})
   \\&+\sum\limits_{t=1}^{k-1}(-1)^{t}\sum_{\sum\limits_{l=0}^{k} m_l^{\hat{t}}=r-k+1}P(F_0^{m_0},\cdots,F_{t-1}^{m_{t-1}},\eta^{t+1,t-1},F_{t+1}^{m_{t+1}},\cdots,F_k^{m_k})
    \\&=\sum_{\sum\limits_{l=0}^{k} m_l^{\hat{0}}=r-k}P(F_1^{m_1+1},\eta^{2,1},F_{2}^{m_{2}},\cdots,F_k^{m_k})+\sum_{\sum\limits_{l=0}^{k} m_l^{\widehat{0,1}}=r-k+1}P(\eta^{2,1},F_{2}^{m_{2}},\cdots,F_k^{m_k})
    \\&+(-1)^k\sum_{\sum\limits_{l=0}^{k} m_l^{\hat{k}}=r-k} P(F_0^{m_0},\cdots,\eta^{k-1,k-2},F_{k-1}^{m_{k-1}+1})+(-1)^k\sum_{\sum\limits_{l=0}^{k-1} m_l^{\widehat{k-1,k}}=r-k+1} P(F_0^{m_0},\cdots,F_{k-2}^{m_{k-2}},\eta^{k-1,k-2})
    \\&+\sum\limits_{t=1}^{k-1}(-1)^{t}\sum_{\sum\limits_{l=0}^{k} m_l^{\hat{t},\widetilde{t+1}}=r-k}P(F_0^{m_0},\cdots,F_{t-1}^{m_{t-1}+1},\eta^{t+1,t-1},F_{t+1}^{m_{t+1}},\cdots,F_k^{m_k})
    \\&+\sum\limits_{t=1}^{k-1}(-1)^{t}\sum_{\sum\limits_{l=0}^{k} m_l^{\widehat{t,t+1}}=r-k}P(F_0^{m_0},\cdots,F_{t-1}^{m_{t-1}+1},\eta^{t+1,t-1},\eta^{t+2,t+1},F_{t+2}^{m_{t+2}},\cdots,F_k^{m_k})
    \\&+\sum\limits_{t=1}^{k-1}(-1)^{t}\sum_{\sum\limits_{l=0}^{k} m_l^{\widehat{t-1,t}}=r-k+1}P(F_0^{m_0},\cdots,F_{t-2}^{m_{t-2}},\eta^{t-1,t-2},\eta^{t+1,t-1},F_{t+1}^{m_{t+1}},\cdots,F_k^{m_k})
\end{align*}

In the above calculation, we mainly decompose $\sum\limits_{l=0}^{k} m_l^{\hat{t}}=r-k+1$ into two main scenarios: $m_{t-1}=0$ and $m_{t-1} \neq 0$. For $m_{t-1} = 0$, we have $\sum\limits_{l=0}^{k} m_l^{\hat{t}}=r-k+1$, which implies $\sum\limits_{l=0}^{k} m_l^{\widehat{t-1,t}}=r-k+1$. For $m_{t-1} \neq 0$, we shift $m_{t-1}$ to $m_{t-1}-1$, so $\sum\limits_{l=0}^{k} m_l^{\widetilde{t-1},\hat{t}}=r-k+1$ becomes $\sum\limits_{l=0}^{k} m_l^{\hat{t}}=r-k$.

Thus, by means of Lemma \ref{Lemma 4.2}, we get relation \eqref{THE}. This completes the proof of the theorem. 
\end{proof}
Above, we discussed Chern-Simons type characteristic classes on the lattice. Now, we turn our attention to their relationship with the continuous counterpart. The following property shows that the $Q$-polynomials \eqref{doc} on the lattice are consistent with the conclusion of the Chern-Simons type characteristic classes in \cite{han1985chern} in the continuous limit.

\begin{proposition}
    With $k+1$ connections $A_0,\cdots,A_k$, we define 
    \begin{align}
        A_{0;t_1\cdots,t_k}=A_0+\sum\limits_{i=1}^{k}t_i\eta^{i,0}, 
    \end{align}
in which $\eta^{i,0}=A_i-A_0$. The curvature of $A_{0;t_1\cdots,t_k}$ is
\begin{align}
     F_{0;t_1 \cdots t_k}=dA_{0;t_1 \cdots t_k}.
\end{align}
Then, in continuous case, we have
     \begin{align}
         &\sum_{\sum\limits_{l=0}^{k} m_l=r-k}P(\eta^{1,0},\cdots,\eta^{k,k-1},F_0^{m_0},\cdots,F_k^{m_k})\nonumber\\
         &=\frac{r!}{(r-k)!} \int_{\Delta_{k}} dt_1 \wedge \cdots \wedge dt_k P(\eta^{1,0},\cdots,\eta^{k,0},F_{0,t_1 \cdots t_k}^{r-k}),
     \end{align}
\end{proposition}

\begin{proof}
    \begin{align}
    F_{0;t_1 \cdots t_k}=dA_{0,t_1 \cdots t_k}=F_0+\sum\limits_{i=1}^{k}t_i d\eta^{i,0}=(1-\sum\limits_{i=1}^{k}t_i)F_0+\sum\limits_{i=1}^{k}t_i F_i=\sum\limits_{i=0}^{k}t_i F_i,
\end{align}
with $t_0=1-\sum\limits_{i=1}^{k}t_i$.
   Thus,
   \begin{align}
    &\frac{r!}{(r-k)!} \int_{\Delta_{k}} dt_1 \wedge \cdots \wedge dt_k P(\eta^{1,0},\cdots,\eta^{k,0},F_{0;t_1 \cdots t_k}^{r-k})\nonumber\\
    &=\frac{r!}{(r-k)!} \int_{\Delta_{k}} dt_1 \wedge \cdots \wedge dt_k P\Big(\eta^{1,0},\cdots,\eta^{k,0},\sum_{\sum\limits_{i=0}^{k}m_i=r-k}\frac{(r-k)!}{m_0 !\cdots m_k !}(t_i F_i)^{m_i}\Big)\nonumber\\
    &=\sum_{\sum\limits_{i=0}^{k}m_i=r-k}\frac{r!}{m_0 !\cdots m_k !}\int_{\Delta_{k}} t_0^{m_0} \cdots t_k^{m_k} dt_1 \wedge \cdots \wedge dt_k   P(\eta^{1,0},\cdots,\eta^{k,0},F_0 ^{m_0},\cdots, F_k ^{m_k})\nonumber\\
    &=\sum_{\sum\limits_{i=0}^{k}m_i=r-k}\frac{r!}{m_0 !\cdots m_k !}\int_{0}^{1}dt_1 \int_{0}^{1-t_1} dt_2 \cdots \int_{0}^{1-t_1-\cdots t_{k-1}} dt_k t_0^{m_0} \cdots t_k^{m_k}  P(\eta^{1,0},\cdots,\eta^{k,0},F_0 ^{m_0},\cdots, F_k ^{m_k})\nonumber\\
    &=\sum_{\sum\limits_{i=0}^{k} m_i=r-k}P(\eta^{1,0},\cdots,\eta^{k,0},F_0^{m_0},\cdots,F_k^{m_k})\nonumber\\
    &=\sum_{\sum\limits_{l=0}^{k} m_i=r-k}P(\eta^{1,0},\cdots,\eta^{k,k-1},F_0^{m_0},\cdots,F_k^{m_k}).
     \end{align}
In the previous derivation, we used the generalized Beta function
 \begin{align}
     B(m_1+1,\cdots,m_k+1)&=\int_{\Delta_{k}} (1-t_1-\cdots-t_k)^{m_0} t_1^{m_1} \cdots t_k^{m_k} dt_1 \wedge \cdots \wedge dt_k\\\nonumber
     &=\frac{\prod _{i=0}^k \Gamma (m_i+1)}{\Gamma (\sum_{i=0}^{k}m_i +k+1)}.
     \end{align}
 $\Gamma(x)$ is the Gamma function, which satisfies $\Gamma(n+1)=n!$  for all $ n \in \mathbb{N}^+$.
\end{proof}
In the case of $k=1$, the definition \eqref{Q} tells us 
\begin{equation}
\left\{\begin{aligned}
   &Q_r^{(0)}(A_0^{(n)};\Delta_0)=P(F_0^{r}),\\
   &\\
  &Q_r^{(1)}(A_0^{(n)}, A_1^{(n)};\Delta_1)=\sum\limits_{m_0=0}^{r-1} P(F_0^{m_0},\eta^{1,0},F_1^{r-1-m_0}).  
\end{aligned}\right.
\end{equation}
Then, \eqref{THE} gives
\begin{align}
   P(F_1^{r})-P(F_0^{r})=dQ_r^{(1)}(A_0^{(n)}, A_1^{(n)};\Delta_1).
\end{align}
When $r=2,A_0^{(n)}=0, A_1^{(n)}=A^{(n)}$, we have
\begin{align}
     P(F^{2})=dQ_2^{(1)}(0, A^{(n)};\Delta_1),\quad Q_2^{(1)}(0, A^{(n)};\Delta_1)=P(A^{(n)},F),
\end{align}
where $P(A^{(n)},F)$ is the abelian Chern-Simons form on the cell $c(n)$.

In the case $k=2$, the theorem\eqref{THE} reads
\begin{align}\label{CS tran}
    Q_r^{(1)}(A_1^{(n)}, A_2^{(n)};\Delta_1)-Q_r^{(1)}(A_0^{(n)}, A_2^{(n)};\Delta_1)+Q_r^{(1)}(A_0^{(n)}, A_1^{(n)};\Delta_1)=dQ_r^{(2)}(A_0^{(n)}, A_1^{(n)},A_2^{(n)};\Delta_2).
\end{align}
The $ Q_r^{(1)}$ are defined as above, and $Q_r^{(2)}(A_0^{(n)}, A_1^{(n)},A_2^{(n)};\Delta_2)$ is
\begin{align}
    Q_r^{(2)}(A_0^{(n)}, A_1^{(n)},A_2^{(n)};\Delta_2) =\sum_{\sum\limits_{l=0}^{2} m_l=r-2}P(F_0^{m_0},\eta^{1,0},F_1^{m_1},\eta^{2,1},F_2^{m_2}).
\end{align}
Consider $r=2, A_0^{(n)}=0, A_1^{(n)}=A^{(n)}, A_2^{(n)}=A^{(n)}-d\lambda(n)-2\pi dN(n)$, such that $N(n) \in \mathbb{Z}$ \cite{fujiwara2001topological}. The \eqref{CS tran} is equal to
\begin{align}
   CS(A_2^{(n)})-CS(A^{(n)})
    &=dP(A^{(n)},d\lambda(n)+2\pi dN(n))+ Q_2^{(1)}(A_1^{(n)}, A_2^{(n)})\\\nonumber
    &=d(P(d\lambda(n),A^{(n)})+2\pi P(dN(n),A^{(n)})).
\end{align}
The $P(d\lambda(n),A^{(n)})$  is the abelian Wess-Zumino action on the cell $c(n)$.

Thus far, we have focused on the Abelian lattice gauge theory. This naturally leads to the question: Can the same construction be applied to the Chern-Simons type characteristic classes in the non-Abelian lattice gauge theory? It is worth pointing out that the second Chern form is no longer a closed form  when we consider non-Abelian gauge field theory. This implies that new methods need to be explored.
\begin{align}
    dP(F,F)&=P(dF,F)+P(F,dF)\\\nonumber
    &=-P(A,F,F)+P(F,A,F)-P(F,A,F)+P(F,F,A)\\\nonumber
    &=-P(A,F,F)+P(F,F,A)\\\nonumber
    &\neq0.
\end{align}

\section{Conclusion and outlook}
In this paper, we construct the $G$-invariant form on the lattice and derive its commutation law, then we define the Chern-Simons type characteristic classes on the lattice. We prove the main result that  the exterior differential of the $k$-th Chern-Simons type characteristic class is exactly equal to 
the coboundary of the cochain of the $(k-1)$-th Chern-Simons type 
characteristic class, and demonstrate that it is consistent with the results from previous work in the continuous limit.

There are many issues worth discussing. Given the development of gauge theory into 2-gauge theory, an interesting question is whether Chern-Simons type characteristic classes can be extended to 2-Chern-Simons classes. Another problem is extending the conclusions of this paper to the 2-lattice gauge theory. Moreover, Chern-Simons type characteristic classes for non-Abelian lattice gauge theory are an important and intriguing issue to explore.

\section *{Acknowledgment}
  
  This work is supported by the National Natural Science Foundation of China (Nos.11971322).
  
  The authors gratefully acknowledge the editors and reviewers for their insightful comments and the contributions made to revising the manuscript.


\begin{thebibliography}{60}
 \bibitem{chern1974characteristic} S. S. Chern, J. Simons,  \href{https://doi.org/10.2307/1971013}{Ann. Math. 99 (1974) 48-69.}
 
 \bibitem{han1985chern} H. Y. Guo, K. Wu, S. K. Wang, \href{https://doi.org/10.1088/0253-6102/4/1/113}{Commun. Theor. Phys. 4 (1985) 113.}
 
 \bibitem{han1985anomalies} H. Y. Guo, B. Y. Hou, S. K. Wang, K. Wu,  \href{https://doi.org/10.1088/0253-6102/4/1/145}{Commun. Theor. Phys. 4 (1985) 145.}
 
 \bibitem{chou2009gauge} K. C. Chou, H. Y. Guo, k. Wu, X. C. Song,  \href{https://doi.org/10.1016/0370-2693(84)90986-9}{Phys. Lett. B 134 (1984) 67-69.}
 
 \bibitem{wess1971consequences} J. Wess, B. Zumino, \href{https://doi.org/10.1016/0370-2693(71)90582-X}{Phys. Lett. B 37 (1971) 95-97.}
 
 \bibitem{witten1983global}  E. Witten, \href{https://doi.org/10.1016/0550-3213(83)90063-9}{Nucl. Phys. B 223 (1983) 422-432.}
 
 \bibitem{Chou_1985} K. C. Chou, H. Y. Guo, K. Wu, \href{https://doi.org/10.1142/9789814280389_0067}{Commun. Theor. Phys. 4 (1985) 91.}
 
 \bibitem{Chou_1984}  K. C. Chou, H. Y. Guo, X. Y. Li, K. Wu, X. C. Song, \href{https://doi.org/10.1088/0253-6102/3/4/491}{Commun. Theor. Phys. 3 (1984) 491.}
 
 \bibitem{chou1984symmetric}  K. C. Chou, H. Guo, K. Wu, X. Song, \href{https://doi.org/10.1142/9789814280389_0063}{Commun. Theor. Phys. 3 (1984) 593.}
 
 \bibitem{kang2018descent}  B. Kang, Y. Pan, K. Wu, J. Yang, Z. F. Yang, \href{https://doi.org/10.1088/0253-6102/69/4/375}{Commun. Theor. Phys. 69 (2018) 375.}
 
 \bibitem{alekseev2018chern}  A. Alekseev, F. Naef, X. M. Xu, C. C. Zhu, Chern-Simons, \href{https://doi.org/10.1007/s11005-017-0985-4}{Lett. Math. Phys. 108 (2018) 757-778.}
 
  \bibitem{izaurieta2015chern} F. Izaurieta, I. Muñoz, P. Salgado, \href{https://doi.org/10.1016/j.physletb.2015.08.030}{Phy. Lett. B 750 (2015) 39-44.}
  
  \bibitem{izaurieta2017chern} F. Izaurieta, P. Salgado, S.Salgado, \href{https://doi.org/10.1016/j.physletb.2017.02.016}{Phys. Lett. B 767 (2017) 360-365.}
  
  \bibitem{danhua2024higher} D. H. Song, K. Wu, J. Yang, \href{https://doi.org/10.1016/j.physletb.2023.138374}{Phys. Lett. B 848 (2024) 138374.}
  
 \bibitem{gromov2015framing} A. Gromov, G. Y. Cho, Y. You, A. G. Abanov, E. Fradkin, \href{https://doi.org/10.1103/PhysRevLett.114.016805}{Phys. Rev. Lett. 114 (2015) 016805.}
 
     \bibitem{niedermayer1999exact}  F. Niedermayer, Exact chiral symmetry, \href{https://doi.org/10.1016/S0920-5632(99)85011-7}{Nucl. Phys. B Proc.
     Suppl. 73 (1999) 105-119.}
 
 \bibitem{luscher1999topology}  M. Lüscher, \href{https://doi.org/10.1016/S0550-3213(98)00680-4}{Nucl. Phys. B 538
 (1999) 515-529.}

\bibitem{fujiwara2001application}  T. Fujiwara, H. Suzuki, K. Wu, in: \href{https://link.springer.com/chapter/10.1007/978-94-010-0704-7_2}{Noncommutative Differential Geometry and Its Applications} \href{https://link.springer.com/chapter/10.1007/978-94-010-0704-7_2}{to Physics, 1999, pp. 13-30.}

\bibitem{zhang2022abelian}  B. Zhang, \href{https://doi.org/10.1103/PhysRevD.105.014507}{Phys. Rev. D 105 (2022) 014507.}

\bibitem{dimakis1993noncommutative}  A. Dimakis, F. Muller-Hoissen, T. Striker, \href{https://doi.org/10.1088/0305-4470/26/8/019}{J. Phys. A: Math. Gen. 26 (1993) 1927.}

\bibitem{fujiwara2000non} T. Fujiwara, H. Suzuki, K. Wu, \href{https://doi.org/10.1016/S0550-3213(99)00706-3}{Nucl. Phys. B 569 (2000) 643-660.}

\bibitem{sulejmanpasic2019abelian}  T. Sulejmanpasic, C. Gattringer, \href{https://doi.org/10.1016/j.nuclphysb.2019.114616}{Nucl. Phys. B 943 (2019) 114616.}

\bibitem{fujiwara2001topological}  T. Fujiwara, H. Suzuki, K. Wu, \href{https://doi.org/10.1103/PhysRevLett.51.638}{Prog. Theor.
Phys. 105 (2001) 789-807.}

\bibitem{connes1994noncommutative} A. Connes, \href{https://alainconnes.org/wp-content/uploads/book94bigpdf.pdf}{Academic Press, New York, 1994.}

\bibitem{hanying2000noncommutative} H. Y. Guo, K. Wu, W. Zhang, \href{https://doi.org/10.1088/0253-6102/34/2/245}{Commun. Theor. Phys. 34 (2000) 245.}

\bibitem{wu2006difference}  K. Wu, W. Z. Zhao, H. Y. Guo, \href{ https://doi.org/10.1007/s11425-006-2060-y}{Sci.
China Ser. A: Math. 49 (2006) 1458-1476.}

     
\end{thebibliography}
\end{document}